%% file: main.tex
\documentclass{tlp}

\usepackage{graphicx}
\usepackage{multirow}
\usepackage{listings}
\usepackage{url}
\usepackage{xcolor}
\usepackage{soul}
\usepackage{tikz}
\usetikzlibrary{arrows.meta, positioning, calc, automata, fit}
\usetikzlibrary{shapes.misc}
\tikzset{cross/.style={cross out, draw=black, minimum size=2*(#1-\pgflinewidth), inner sep=0pt, outer sep=0pt},
cross/.default={1pt}}

\usepackage{booktabs}
\usepackage{pgfplots}
\usepackage{cleveref}
\crefname{equation}{Eqn.}{Eqn.}
\crefname{figure}{Fig.}{Fig.}
\crefname{table}{Table}{Table}
\crefname{section}{Sec.}{Sec.}
\usepackage{subcaption}
\usepackage{comment}
\pgfplotsset{compat=1.18}
\usepackage{float}
\usepackage{amssymb}
\usepackage{mathtools}

\lstdefinelanguage{souffle}{
  sensitive=true,
  morecomment=[l]{//},
  morecomment=[s]{/*}{*/},
  morestring=[b]",
  alsoletter={.:-_},
  keywords=[1]{.decl,.input,.output,.type,.number,.symbol,.pragma,.include},
  keywords=[2]{rule,if,then,else,true,false},
  keywords=[3]{and,or,not,exists},
  keywords=[4]{count,sum,min,max,mean},
  keywordstyle=[1]\color{blue}\bfseries,
  keywordstyle=[2]\color{teal}\bfseries,
  keywordstyle=[3]\color{purple}\bfseries,
  keywordstyle=[4]\color{orange}\bfseries,
  basicstyle=\ttfamily\footnotesize,
  columns=fullflexible,
  keepspaces=true,
  showstringspaces=false,
  breaklines=true,
  breakatwhitespace=true,
  tabsize=2,
}

\lstdefinestyle{souffleStyle}{
  language=souffle,
  numbers=left,
  numberstyle=\tiny\color{gray},
  stepnumber=1,
  numbersep=8pt,
  frame=single,
  rulecolor=\color{black!20},
  captionpos=b,
}

\newcommand{\dtinline}[1]{\lstinline[style=souffleStyle]!#1!}

\newtheorem{definition}{Definition}

\newtheorem{invariant}{Invariant}
\newtheorem{lemma}{Lemma}

\newcommand{\F}[1]{\mathcal{F}_{\text{#1}}}
\newcommand{\visrel}{\xrightarrow{\mathrm{vis}}}

\newcommand{\add}[1]{\texttt{add}(#1)}
\newcommand{\del}[1]{\texttt{del}(#1)}
\newcommand{\hasV}[1]{\texttt{hasV}(#1)}

\newcommand{\upd}[2]{\texttt{upd}(#1,#2)}
\newcommand{\rmv}[1]{\texttt{rmv}(#1)}
\newcommand{\hasKV}[2]{\texttt{hasKV}(#1,#2)}

\newcommand{\addN}[1]{\texttt{addN}(#1)}
\newcommand{\rmvN}[1]{\texttt{rmvN}(#1)}
\newcommand{\addE}[2]{\texttt{addE}(#1,#2)}
\newcommand{\rmvE}[2]{\texttt{rmvE}(#1,#2)}
\newcommand{\hasN}[1]{\texttt{hasN}(#1)}
\newcommand{\hasE}[2]{\texttt{hasE}(#1,#2)}

\newcommand{\idg}{\texttt{IDGraph}}
\newcommand{\ddg}{\texttt{DDGraph}}
\newcommand{\nodeset}{\texttt{NodeSet}}
\newcommand{\edgemap}{\texttt{EdgeMap}}
\newcommand{\valset}{\texttt{ValSet}}
\newcommand{\keyset}{\texttt{KeySet}}
\newcommand{\edgeset}{\texttt{EdgeSet}}
\newcommand{\mapset}{\texttt{MapSet}}

\newcommand{\vis}{\text{vis}}
\newcommand{\E}{\text{E}}
\newcommand{\op}{\text{op}}
\newcommand{\ar}{\text{ar}}

\newcommand{\AB}[1]{\textcolor{purple}{AB: #1}}

\newcommand{\Revision}[1]{\textcolor{blue}{#1}}

\begin{document}

\lefttitle{Cambridge Author}

\jnlPage{1}{8}
\jnlDoiYr{2021}
\doival{10.1017/xxxxx}

\title[A Datalog Framework for Conflict-Free Replicated Data Types]
{A Datalog Framework for \\Conflict-Free Replicated Data Types}

\begin{authgrp}
\author{Elena Yanakieva}
\affiliation{RPTU University Kaiserslautern-Landau}

\author{Annette Bieniusa}
\affiliation{RPTU University Kaiserslautern-Landau}

\author{Stefania Dumbrava}
\affiliation{ENSIIE, Inria Paris, IRIF, SAMOVAR (Télécom SudParis)}
\end{authgrp}


\maketitle

\begin{abstract}

Distributed applications increasingly support local-first collaboration over shared data, allowing multiple users to perform updates concurrently without global coordination. Such collaboration requires careful design to capture the intended semantics of the concurrent interactions.

We introduce a declarative framework for specifying and reasoning about the semantics of conflict-free replicated data types (CRDTs) and CRDT-based applications in Datalog. The framework models CRDT semantics as executable logic programs over operation contexts, making concurrency explicit and compositional, and thus amenable to automated analysis. As one application, we use property-based testing to compare implementations. To the best of our knowledge, this is the first work to systematically use Datalog as a foundation for prototyping and analyzing complex CRDTs and their compositions.

We evaluate our methodology using a collaborative graph data editing case study and report experimentation results assessing correctness validation and scalability with an increasing number of operations and replicas.
\end{abstract}


\section{Introduction}

\input{introduction_new}


\section{Background: Collaborative Data Structures} \label{background}

\input{background}


\section{A Declarative Framework for CRDT Specification} \label{spec}
\input{implementation}

\section{A Case Study on Graph CRDTs}\label{graph}
\input{use-case}

\section{Experimental Evaluation}\label{experiments}
\input{experiments}

\section{Conclusions and Perspectives}
\input{conclusions}

\bibliographystyle{plainnat}
\bibliography{bibliography}
\newpage
\input{appendix}

\end{document}

%% file: introduction_new.tex
Distributed applications increasingly support collaboration over shared data, allowing multiple users to update the same data concurrently. This raises a fundamental question: when concurrent updates conflict, to which state should the application converge? To ensure predictable and well-defined behavior, developers must therefore specify application-specific conflict semantics. From a semantic perspective, the core challenge is, thus, not only convergence itself but also the lack of explicit, analyzable specifications that explain why a particular concurrent outcome arises.

This challenge is amplified by the need for high availability. 
To enable offline work and low-latency local updates, collaborative systems must weaken consistency guarantees \citep{capTheorem} and reconcile concurrent updates afterwards. 
Conflict-free Replicated Data Types (CRDTs) address this need by providing replicated data structures with deterministic convergence and strong eventual consistency (SEC)~\citep{crdts}.
However, building applications from CRDTs remains difficult in practice. 
Developers typically combine multiple data structures, but conflict-resolution policies do not compose predictably. 
Subtle concurrency patterns can cause otherwise correct CRDTs to converge to states that violate the intended application semantics. While CRDTs guarantee deterministic convergence, they do not guarantee that composed CRDTs converge to a state consistent with the intended application semantics.

Our work addresses this problem by providing a declarative modeling framework and tool that supports developers in exploring the composability of CRDT semantics.

\paragraph{Related Work.} 
Dedalus (\cite{dedalus}) introduces a temporal Datalog-based formalism for reasoning about replicated mutable state and asynchronous communication in long-running distributed systems, whereas our focus is on executable semantic specifications for replicated data types and their compositions. 
Building on these ideas, the Bloom language~(\cite{bloom}) proposes a declarative framework for reasoning about distributed programs through monotonic programming over sets, assuming a fixed merge operator (union) and a containment-based partial order.

Subsequent systems explored deterministic and coordination-free programming models using lattice-based structures. LVars~(\cite{KuperN13}) introduced deterministic-by-construction parallelism over lattices, while Lasp~(\cite{MeiklejohnR15}) proposed a functional programming model for composing state-based CRDTs into larger computations that satisfy SEC. More recently, the CALM theorem~(\cite{calm}) established that logically
monotonic programs achieve consistency without coordination, reinforcing
Datalog as a natural foundation for such reasoning.



The Peepul approach~(\cite{soundarapandian2022certified}) certifies MRDTs with machine-checked convergence proofs via F* and SMT solving, targeting individual CRDT correctness rather than application-level semantics across compositions. Their approach complements ours: while they derive verified implementations, we make semantic mismatches between intended behavior and CRDT compositions explicit, executable, and falsifiable.

The work in~\cite{DBLP:conf/papoc/PandeyD00PP25} analyzes the challenges of applying replicated data types to property graph databases in local-first settings, identifying how graph-specific constraints interact with concurrent updates; our framework provides a concrete methodology for specifying and testing such interactions.

\paragraph{Our approach.} To avoid unintended behavior, developers need to be able to state the intended semantics of an application independently of a particular CRDT design. We propose a workflow in which developers first write a declarative, executable reference specification of the application semantics and then construct an implementation via CRDT composition. 
Our framework supports systematic testing: given generated concurrent executions, it checks whether a composite realization converges to the same state as the reference specification using property-based testing. While formal verification of distributed applications remains expert-level and costly, executable specifications enable lightweight, systematic testing readily usable during design.

Datalog is central to our work. First, it provides a concise and declarative way to express CRDT semantics as relations over events and visibility. Second, its monotonic, bottom-up evaluation model aligns naturally with CRDT semantics, where state is derived deterministically from a growing set of observed events. Third, Datalog specifications are executable, enabling systematic testing without committing to a concrete distributed implementation.

We implement the framework in \textsc{CRDTLog}, a Datalog-based system built on the Soufflé engine. \textsc{CRDTLog} provides a library of declarative specifications for basic, nested, and composite CRDTs, as well as tools to generate operation contexts and compare semantic outcomes. \emph{To the best of our knowledge, this is the first work to systematically use Datalog as a foundation for prototyping and analyzing complex CRDTs and their compositions.}
We make the following contributions:

\begin{itemize}
    \item We present an extended declarative framework for specifying the semantics of CRDTs and collaborative applications based on operation contexts. 
    \item We implement the framework in Soufflé and provide a library of classical, nested, and composite CRDT specifications, including two CRDT-based composite realizations of collaborative graph applications.
    \item We present a case study on collaborative graphs that demonstrates how abstract semantic specifications support CRDT composition choices and expose non-obvious concurrency behaviors.
    \item We evaluate the approach empirically, assessing semantic equivalence between abstract specifications and composite realizations across a set of generated executions.
\end{itemize}

Although our prototype is implemented using the Soufflé Datalog engine, the approach itself is not tied to a specific system. CRDTlog relies on language features, such as support for recursion, stratified negation, and aggregates which are provided by a wide range of Datalog engines (e.g., Vadalog, Nemo, Logica) as well as Prolog systems with tabling. Soufflé was chosen for its performance and mature tooling, but any engine integrating these features can serve as a backend for our framework. To provide a better playground for designers, we implemented our framework also in DDlog (\cite{ddlog}), an incremental Datalog engine that allows users to interactively test use cases.
A formal user study is beyond the current scope and is deferred to future work.

Additional proofs, and experimental details are provided in the supplemental material.

%% file: background.tex
Our work is motivated by \emph{collaborative applications} that require data structures to be shared and updated by multiple users concurrently under weak connectivity and without global coordination. 
As an example, consider a (directed) graph $(V,A)$ where $V$ is a set of nodes and
$A \subseteq V \times V$ is a set of directed edges, as used in a collaborative design canvas or shared knowledge bases. 
Users may modify the graph by adding (and removing) nodes and edges; to simplify the presentation, we focus here on the structural updates and do not consider modifications of labels.
\Cref{tab:ID-graph} summarizes the semantics of such a graph with \emph{isolate-delete (ID)} semantics.

\begin{table}[h]
\centering
\footnotesize
\begin{tabular}{p{1.4cm}p{5cm}p{3.5cm}p{3cm}}
\hline
\textbf{Operation} & \textbf{Precondition} & \textbf{Postcondition} & \textbf{Return value} \\ \hline

\hasN{v}
& - 
& $V' = V \land A' = A$ 
& $v \in V$ \\

\hasE{u}{v} 
& - 
& $V' = V \land A' = A$ 
& $(u,v) \in A$ \\

\addN{v} 
& - 
& $V' = V \cup \{v\} \land A' = A$ 
& void \\ 

\addE{u}{v} 
& $u \in V \land v \in V$ 
& $A' = A \cup \{(u,v)\} \land V' = V$ 
& void \\ 

\rmvE{u}{v}
& - 
& $A' = A \setminus \{(u,v)\} \land V' = V$ 
& void \\ 

\rmvN{v}
& $v \in V \land \forall x \in V: (x,v) \notin A \land (v,x) \notin A$ 
& $V' = V \setminus \{v\} \land A' = A$ 
& void \\ \hline

\end{tabular}
\caption{Operations for a directed graph with isolate-delete (ID) semantics; only isolated, non-connected nodes may be removed.}
\label{tab:ID-graph}
\end{table}

Additionally, we require an application-level invariant, which guarantees the absence of dangling edges: $\forall (u,v)\in A:\ u\in V \land v\in V$.

\begin{figure}[h]
    \begin{subfigure}[b]{0.30\linewidth}
    \centering
    \caption{User A inserts an edge.}
    \vspace{1mm}
    \scalebox{0.8} {
        \begin{tikzpicture}
            \node (q0) [state] {$n$};
            \node (q1) [state, right = of q0] {$m$};
            \path[->]
                (q0) edge node [below] {} (q1);
        \end{tikzpicture}}
    \end{subfigure}
    \begin{subfigure}[b]{0.30\linewidth}
    \caption{User B removes node $m$.}
    \vspace{1mm}
    \centering\scalebox{0.8} {
        \begin{tikzpicture}
            \node (q0) [state] {$n$};
            \node (q2) [state, right = of q0] {};
            \node (q1) [state, right = of q0, cross out] {$m$};
        \end{tikzpicture}}
    \end{subfigure}
    \begin{subfigure}[b]{0.30\linewidth}
        \centering
        \caption{After sync.}
        \vspace{1mm}
        \scalebox{0.8} {
        \begin{tikzpicture}
            \node[state] (q0) {$n$};
            \path[->,draw = red]
                (q0) edge node [below, text=red] {} +(1.5,0);
        \end{tikzpicture}}
    \end{subfigure}
    \vspace{2mm}
    \caption{Dangling edge scenario.}
    \label{figureDanglingEdge}
\end{figure}
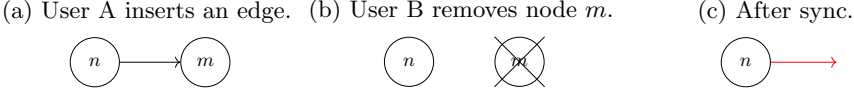

In a single-user, sequential setting, this invariant can be enforced directly by the operations. 
However, under concurrent updates by multiple users, preconditions are not stable, leading to invariant violations.
We illustrate such a violation in \Cref{figureDanglingEdge}.
Consider a graph with two nodes $n$ and $m$. User A executes \addE{n}{m} while user B concurrently executes \rmvN{m}. 
To avoid an invalid graph state, a conflict-resolution strategy is required to preserve the invariant. One possible outcome is to remove $m$ and discard the edge; another is to discard the removal of $m$. No resolution is universally correct: the appropriate behavior depends on the intended application semantics.

Developers must therefore specify how conflicting operations should be interpreted and combined while preserving availability to enable offline execution and immediate local reasoning. The CAP theorem states that, in the presence of network partitions, strong consistency and availability cannot be guaranteed simultaneously (\cite{capTheorem}). Consequently, local-first collaborative systems require an explicit and declarative choice of \emph{conflict semantics}.

Conflict-free Replicated Data Types (CRDTs) provide a standard approach to implementing highly available replicated state with deterministic convergence \citep{crdts, DBLP:conf/papoc/BorthLB25}. Informally, CRDTs guarantee strong eventual consistency (SEC): once all replicas have received the same set of updates, they converge to the same state. Several systems expose CRDTs to application developers, including Yjs \citep{Nicolaescu2015Yjs} and Automerge \citep{Automerge}. These frameworks provide a collection of basic replicated data types, such as sets, maps, sequences, counters, and registers, each equipped with a specific conflict-resolution policy.

Throughout this paper, we refer to common policies for datatypes such as:

\begin{itemize}

\item \emph{add-wins} sets: concurrent \add{v} and \del{v}
leave \texttt{v} present;
\item \emph{update-wins} maps: concurrent \upd{k}{v} and \rmv{k}
preserve the binding for \texttt{k} when the entry is concurrently updated.
\end{itemize}

These policies are useful primitives, but they are not neutral; each one encodes a specific choice about user intent under concurrency.

Real applications rarely consist of a single basic CRDT. Instead, developers compose and nest CRDTs to obtain richer behavior. However, conflict resolution does not automatically compose intuitively \citep{kuessner2023algebraic}: combining CRDTs can introduce unintended behaviors that may only surface under subtle concurrent execution patterns.

This motivates a tool-supported workflow in which developers specify application semantics abstractly, explore alternative CRDT decompositions, and verify that the chosen design matches the intended behavior. Ideally, this relies on an executable specification enabling systematic exploration of concurrent user interactions. The remainder of the paper develops such a framework and implements it in Soufflé~\citep{DBLP:conf/cav/JordanSS16}.

Our framework builds on the declarative model of replicated data types by \citet{burckhardt2014principles}, which defines each data type as a deterministic function over an operation context, in which events and their relation are captured. Note that we generalize the definition of an operation context such that each event can be mapped to a set of operations to model that the atomic execution of all operations of one event as needed in \Cref{souffleCRDTApplications}.

\begin{definition}[Operation Context and Replicated Data Type]
Let $\E$ be a set of events, $Op_{X}$ be the set of operations for a replicated data type $X$ and $\mathsf{Values}$ be a set of possible return values. An \emph{operation context} is a finite event graph $C = (\E,\op,\vis,\ar)$, where $\op: \E \xrightarrow{} \mathcal{P} (Op_X)$ is the set of operations of each event. The operations in the set are required to be commutative and associative so that the order of execution does not matter. $vis$ is an acyclic relation representing visibility among the elements of $\E$, and $\ar$ is a total order representing the arbitration of the elements in $E$. 

A \emph{replicated data type} $\F{X} : Op_X \times C \xrightarrow{} \mathsf{Values}$ is a (deterministic) function that, given an operation $o$ and an operation context $C$, specifies the expected return value $\F{X}(o,C)$ when performing operation $o$ in context $C$. 
\end{definition}

$\F{X}$ is a declarative specification that maps an operation context to the result of each operation.
Unlike \cite{burckhardt2014principles}, our model restricts contexts to update operations.
Queries (e.g., \hasN{v}) are treated as post-hoc observers evaluated over the context to obtain the final state.
Although concurrency is modeled via the visibility relation, the evaluation of $\F{X}$ is deterministic.

In an add-wins set, an element $v$ is present if there exists an $\add{v}$ event $e_1$ and no $\del{v}$ event $e_2$ that observes it (i.e., $e_1$ is not visible to $e_2$). Thus, a delete concurrent with the add does not remove $v$, capturing add-wins semantics. Formally:
\begin{equation}
{\footnotesize
\begin{aligned}
&\F{SetAW}\big(\hasV{v}, (\E, \op, \vis, \ar)\big) = \text{true} \\
&\iff \exists e \in E: \add{v} \in \op(e) \land \neg\big(\exists e' \in E: \del{v} \in \op(e') \land e \visrel e' \big)
\end{aligned}
}
\label{eq:awset}
\end{equation}

\Cref{fig:opcontext} illustrates an example operation context for an add-wins set. 
The context contains four events: event 1 performs \add{$a$}, event 2 \del{$a$}, event 3 \add{$a$}, and event 4 performs \add{$b$}. The visibility relation includes edges 
$1 \to 2, 3 \to 4$ and $2 \to 4$, meaning that event 2 observes event 1; event 4 observes events 3, 2, and transitively also 1.
Event 3 is concurrent with events 1 and 2.
As arbitration order is irrelevant for an add-wins set, it is omitted. The resulting set contains $a$ and $b$, since the concurrent \add{$a$} in event 3 overrules \del{$a$} in event 2.

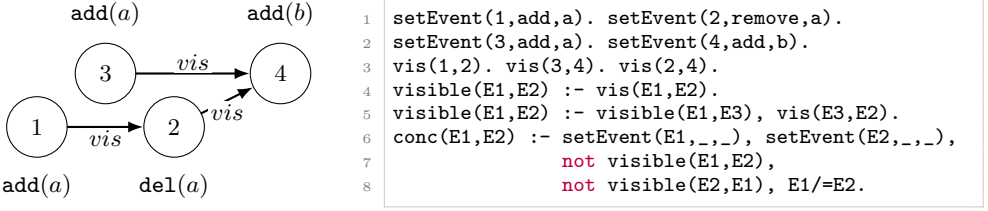
\begin{figure}[h]
\centering

\begin{subfigure}[t]{0.40\textwidth}
\vspace{0pt}
\centering
\begin{tikzpicture}[
  nodeStyle/.style={circle, draw, minimum size=8mm, inner sep=2pt},
  opStyle/.style={font=\small},
  labStyle/.style={font=\small, fill=white, inner sep=1pt},
  edgeStyle/.style={thick, -{Latex[length=2mm]}, line cap=round, line join=round},
]
\node[nodeStyle] (n1) {1};
\node[nodeStyle, right=1cm of n1] (n2) {2};
\node[nodeStyle, above=0.3cm of $(n1)!0.5!(n2)$] (n3) {3};
\node[nodeStyle, right=1.5cm of n3] (n4) {4};

\node[opStyle, below=1mm of n1] {\add{$a$}};
\node[opStyle, below=1mm of n2] {\del{$a$}};
\node[opStyle, above=1mm of n3] {\add{$a$}};
\node[opStyle, above=1mm of n4] {\add{$b$}};

\draw[edgeStyle] (n1) -- node[labStyle, below] {$vis$} (n2);
\draw[edgeStyle] (n2) -- node[labStyle, below] {$vis$} (n4);
\draw[edgeStyle] (n3) -- node[labStyle, above] {$vis$} (n4);
\end{tikzpicture}
\end{subfigure}\hfill
\begin{subfigure}[t]{0.57\textwidth}
\vspace{0pt} 
\raggedright 

\begin{lstlisting}[style=souffleStyle, linewidth=\linewidth]
setEvent(1,add,a). setEvent(2,remove,a).
setEvent(3,add,a). setEvent(4,add,b).
vis(1,2). vis(3,4). vis(2,4).
visible(E1,E2) :- vis(E1,E2).
visible(E1,E2) :- visible(E1,E3), vis(E3,E2).
conc(E1,E2) :- setEvent(E1,_,_), setEvent(E2,_,_), 
               not visible(E1,E2), 
               not visible(E2,E1), E1/=E2.
\end{lstlisting}

\end{subfigure}

\caption{Example operation context for an add-wins set and its Datalog encoding.}
\label{fig:opcontext}
\end{figure}

%% file: implementation.tex
We introduce \textsc{CRDTlog}\footnote{\textsc{CRDTLog}'s code is available at \url{https://github.com/elly-yanakieva/CRDTLog}.}, an implementation of the declarative approach to CRDTs from Sec.~\ref{background}. It extends the framework with nested and composite CRDTs, enabling designers to model CRDT-based applications and reason about their semantics. This approach is generic: users can define arbitrary CRDTs with custom conflict resolution strategies. For instance, a delete-wins set, where any observed deletion takes precedence, can be captured by the specification in \Cref{eq:dwset}.
\begin{equation}
\footnotesize
\F{SetDW}\big(\hasV{v}, (\E, \op, \vis, \ar)\big) = \text{true} \iff \neg\big(\exists e \in E: \del{v} \in \op(e))
\label{eq:dwset}
\end{equation}
\textsc{CRDTlog} is implemented in Soufflé, a Datalog engine with efficient, scalable bottom-up evaluation~\citep{DBLP:conf/cav/JordanSS16}. For readability, we present the theory and its Datalog implementation side by side: we represent operation contexts as input relations and derive return values (for queries) as output relations.

Our framework provides a “playground” for the design phase (\Cref{fig:workflow}). Developers first write the \emph{specification-level semantics (SLS)}, a high-level specification of the intended application semantics, then propose the \emph{implementation compositional semantics (ICS)}: a CRDT composition with transformation rules mapping application operations to component operations. Both are executable over generated operation contexts, enabling systematic exploration of corner cases and property-based equivalence checks. If equivalence fails, the decomposition is revised and the process iterates. Crucially, the framework isolates the semantic core, i.e., how conflicts are resolved and how composed CRDTs interact, while abstracting away protocol and engineering choices, which can be deferred until after the semantic design is validated.

Throughout this section, we refer to the collaborative graph introduced in Sec.~\ref{background} and summarized in Table \ref{tab:ID-graph} as a running example to illustrate the workflow. The operations of the used datatypes can be found in \Cref{tab:operations}.

\begin{table}[t]
\centering
\footnotesize
\begin{tabular}{p{1.6cm}p{6.5cm}p{5cm}}
\textbf{Datatype} & \textbf{Update operations} & \textbf{Query operations}\\ 

Set 
& \add{$v$}, \del{$v$} 
& \hasV{$v$} \\ 

Map$\langle$Set$\rangle$ 
& \upd{$k$}{\add{$v$}}, \upd{$k$}{\del{$v$}}, \rmv{$k$}
& \hasKV{$k$}{$v$}\\ 

Graph 
& \addN{$n$}, \rmvN{$n$}, \addE{$n$}{$m$}, \rmvE{$n$}{$m$}  
& \hasN{$n$}, \hasE{$n$}{$m$}
\end{tabular}
\vspace{1mm}
\caption{Summary of datatype operations.}
\label{tab:operations}
\end{table}

\input{workflow_diagram_v2}

\subsection{Operation Contexts in Datalog}

We encode an operation context in \textsc{CRDTLog} using three input relations: one for events and their operations, \dtinline{vis} for the visibility relation, and \dtinline{ar} for the arbitration order.  

CRDT systems typically derive visibility, a.k.a. happens-before, from metadata such as vector clocks, and arbitration order from timestamps. 
The arbitration order is used only when semantics require a total order on concurrent events (e.g., last-writer-wins); otherwise, conflict resolution relies solely on visibility.

We illustrate the operation context of the add-wins set from Fig.~\ref{fig:opcontext}. 
Event facts in Datalog refer to an event identifier, the operation name, and its parameters. 
For example, event \dtinline{1} applies \dtinline{add} with parameter \dtinline{a}.
For readability, the input relation \dtinline{vis} contains only direct visibility edges, and we compute its transitive closure to obtain full visibility. 
Events are considered concurrent when they are not connected by the visibility relation.

In general, CRDT correctness is argued via \emph{strong eventual consistency}: replicas that have (eventually) observed the same set of updates must converge to the same state. 
In our setting, concurrency and update propagation is captured explicitly in the operation context which we assume to be well-formed and identical across replicas once they have observed the same updates. 
Each data type $X$ is then defined by function 
$\F{X}$ over this context. 
Consequently, any two replicas that evaluate $\F{X}$ on the same operation context derive the same output relations, yielding convergence by construction. 
Since our Datalog programs implement $\F{X}$ as rule evaluation over the relations representing the context, they inherit this determinism directly.

\subsection{CRDT Encoding in Datalog}
In Section 2, we introduced declarative semantics for basic CRDTs, such as the add-wins set (\Cref{eq:awset}). We express these semantics as Datalog rules that derive output relations representing the observable state at the end of the operation context. 

\begin{lstlisting}[style=souffleStyle, caption = Add-wins set semantics., label=lst:awset]
postVisOp(E1) :- setEvent(E1, _, X), setEvent(E2, "del", X), visible(E1, E2), E1 \= E2.
setState(X)   :- setEvent(E, "add", X), not postVisOp(E).
\end{lstlisting}



\Cref{lst:awset} gives the Datalog encoding of the add-wins set.
Events are represented by \dtinline{setEvent(E,O,X)}, where \dtinline{E} is an event identifier, \dtinline{O} an operation, and \dtinline{X} a value.
The predicate \dtinline{setState(X)} computes \hasV{\dtinline{X}} over the operation context.
A value \dtinline{X} is in the set if and only if there exists an \add{\dtinline{X}} that is not canceled by a later \del{\dtinline{X}} (\dtinline{not postVisOp}).

Collaborative applications typically compose basic or nested CRDTs into composite designs~\citep{kuessner2023algebraic}. 
Accordingly, \textsc{CRDTlog} supports nested CRDTs such as maps; our graph case study uses a Map$\langle$Set$\rangle$.

\subsection{Nested CRDTs}
\label{appendix:nested}

Designers obtain nested data types by choosing component CRDTs and writing transformation rules from top-level events to component events. In nested structures, the decomposition is often structural. For example, a Map$\langle$Set$\rangle$  combines a set that tracks keys and a set of values for each key.

To make this example concrete, consider a Map$\langle$Set$\rangle$  with $Op_X = \{\upd{k}{v}, \rmv{k}\}$, where $v$ can be either $\add{x}$ or $\del{x}$. The first step is to define the SLS that formalize the behavior of the application as a single object:
\begin{equation}
{\footnotesize
\begin{aligned}
&\F{MapSet\_SLS}\big(\hasKV{k}{x}, (\E, \op, \vis, \ar)\big) = \text{true} \\
&\iff \exists e \in \E: \upd{k}{\add{x}} \in \op(e) \land \big( \neg \exists e' \in \E: \rmv{k} \in \op(e') \land e \xrightarrow{vis} e' \big) \\
& \quad \land \big( \neg \exists e'' \in \E: \upd{k}{\del{x}} \in \op(e'') \land e \xrightarrow{vis} e'' \big) 
\end{aligned}
}
\label{eq:mapsetadt}
\end{equation}

A Map$\langle$Set$\rangle$  contains a key \dtinline{k} and value \dtinline{x} in the set associated with \dtinline{k} if it was inserted by an \dtinline{upd} operation and neither the key nor the value were later removed.

Next, we obtain the ICS for the Map$\langle$Set$\rangle$. Intuitively, the map contains $(k,x)$ when $k$ is present in the key set and $x$ is present in the value set associated with $k$. \Cref{fig:layered_translation_mapset} presents the corresponding transformation rules.
\begin{equation}
{\footnotesize
\begin{aligned}
&\F{MapSet\_ICS}\big(\hasKV{k}{x}, (\E, \op, \vis, \ar)\big) = \text{true} \iff \\
&\quad \F{Set}\big(\hasV{k}, (\E, \op_{\keyset}, \vis, \ar)\big)
\land  \mathcal{F}_{\text{set}_k}\big(\hasV{x}, (\E_{\valset}, \op_{\valset}, \vis, \ar)\big)
\end{aligned}
}
\label{eq:mapsetics}
\end{equation}
\begin{figure}[h]
\centering
{\footnotesize
\[
\renewcommand{\arraystretch}{1.15}
\begin{array}{rcl}
\mapset.\upd{k}{\add{v}}
&\overset{\mathcal{T}}\Rightarrow&
\{\keyset.\add{k}, \valset.\add{v}\} \\

\mapset.\upd{k}{\del{v}}
&\overset{\mathcal{T}}\Rightarrow&
\{\keyset.\add{k}, \valset.\del{v}\} \\

\mapset.\rmv{k}
&\overset{\mathcal{T}}\Rightarrow& \{ \keyset.\del{k} \} 
\end{array}
\]
}
\caption{Layered transformation for update-wins Map$\langle$Set$\rangle$.}
\label{fig:layered_translation_mapset}
\end{figure}

In Datalog, the event facts of the Map$\langle$Set$\rangle$ \dtinline{mapSetOp(E,Op,InnerOp,K,X)} consist of the event identifier \dtinline{E}, the map-level operation \dtinline{Op}, the inner set operation \dtinline{InnerOp} and the operation parameters \dtinline{K} and \dtinline{X} (Listing \ref{lst:mapset}). \dtinline{keysEvent} and \dtinline{setEvent} correspond to the transformation rules. \dtinline{mapState} computes the final state. Note that a key might exist with an empty value set. For that we use the empty value \dtinline{""}.

\begin{lstlisting}[style=souffleStyle, caption = Map$\langle$Set$\rangle$  semantics expressed in Datalog., label = lst:mapset]
keysEvent(E, "add", K) :- mapSetEvent(E, "upd",  _, K, _).
keysEvent(E, "del", K) :- mapSetEvent(E, "rmv", _, K, _).
setEvent(E, InnerOp, X, K) :- mapSetEvent(E, "upd", InnerOp, K, X), not postVisRemove(E).
mapState(K, X)  :- keysState(K), setState(K, X).
mapState(K, "") :- keysState(K), not setState(K, _).
\end{lstlisting}

Soufflé does not support dynamic instantiation of modules at runtime. In particular, a Map$\langle$Set$\rangle$ cannot allocate a fresh set instance per key during evaluation. We therefore encode all inner sets within a single indexed set instance, using an additional identifier (the key) to distinguish per-key state. The corresponding code can be found in \Cref{lst:valueset}. With this encoding pattern, deeper nesting (e.g. maps of maps) reduces to the repeated use of transformation rules and identifiers.

\begin{lstlisting}[style=souffleStyle, caption = Value set semantics expressed for ICS Map$\langle$Set$\rangle$., label = lst:valueset]
laterVisOpSet(E1, E2) :- setEvent(E1, _, X, S), setEvent(E2, "del", X, S), 
                         visible(E1, E2), E1\=E2.
setState(S, X)        :- setEvent(E, "add", X, S), not laterVisOpSet(E, _).
\end{lstlisting}

\subsection{Composite CRDTs and CRDT Applications} \label{souffleCRDTApplications}

Collaborative applications typically combine multiple CRDTs into composite ones~\citep{kuessner2023algebraic}. 
We illustrate the workflow on composite CRDTs using the ID graph application from Table \ref{tab:ID-graph}. 
Recall that a node is in the graph if it has been added and not later removed, unless a concurrent \dtinline{addE} defeats the removal. 
An edge is present if and only if it has been added and not subsequently removed.
We define the corresponding specification-level semantics (SLS) in \Cref{eq:graphADT}.
\begin{equation}
\label{eq:graphADT}
{\footnotesize
\begin{aligned}
&\F{IDGraph\_SLS}\big(\hasN{n},\E,\op,\vis,\ar\big)=\text{true} \iff \exists e\in \E:\ \addN{n} \in \op(e) \; \\ 
&\land \; \Big(\neg\exists e'\in \E:\ \rmvN{n} \in \op(e')\land e\visrel e'\lor\ \exists e',e''\in \E:\ \rmvN{n} \in \op(e') \land e\visrel e'\\
& \hspace{1.3cm} \land \big(\addE{n}{m} \in \op(e'')\ \lor\ \addE{m}{n} \in \op(e'')\big)\ \land e'\not\visrel e'' \land e''\not\visrel e'\Big)
\\
& \F{IDGraph\_SLS}\big(\hasE{n}{m}, \E, \op, \vis, \ar\big) = \text{true} \\
& \iff \exists e \in \E : \addE{n}{m} \in \op(e)   \land \Big( \neg\exists e' \in \E : \rmvE{n}{m}  \in \op(e') \land e \visrel e' \Big)
\end{aligned}
}
\end{equation}
In the executable SLS in Datalog (Listing \ref{lst:graphidadt}), we model the input as facts \dtinline{graphEvent(E, O, V1, V2)}, where \dtinline{E} is the event identifier, \dtinline{O} is the operation of the event, and \dtinline{V1} and \dtinline{V2} are the parameters of the operation.
For a node \dtinline{n} to be in the resulting graph state \dtinline{nodes(K)}, there must be an \dtinline{addN(n)} which is not visible to a \dtinline{remove} (\dtinline{not postVisRmvN}); or, if there is a \dtinline{postVisRmvN}, there exists a concurrent \dtinline{addE} (\dtinline{concAddE}), outgoing or incoming from \dtinline{n}. An edge is in the resulting graph state (\dtinline{edges(N,M)}) if the edge has been added and it has not been removed later (\dtinline{postVisRmvE}).
\begin{lstlisting}[style=souffleStyle, caption = SLS semantics of isolate-delete graph., label = lst:graphidadt]
conc(E1,E2)     :- graphEvent(E1,_,_,_), graphEvent(E2,_,_,_), not visible(E1,E2), 
                   not visible(E2,E1), E1\=E2.
concAddE(E2, K) :- graphEvent(E2,_,_,_), graphEvent(E1, "addE", K, _), conc(E1,E2).
concAddE(E2, K) :- graphEvent(E2,_,_,_), graphEvent(E1, "addE", _, K), conc(E1,E2).
postVisRmvN(E1, K, E2) :- graphEvent(E1, _, _, _), graphEvent(E2, "rmvN", K, _), 
                          visible(E1, E2), E1\=E2, not concAddE(E2, K).
nodes(K)        :- graphEvent(E, "addN", K, _), not postVisRmvN(E, K, _).
postVisRmvE(E1) :- graphEvent(E1,_,K1, K2), graphEvent(E2, "rmvE", K1, K2), 
                   visible(E1, E2), E1 \= E2.
edges(N, M)     :- graphEvent(E, "addE", N, M), not postVisRmvE(E).
\end{lstlisting}
The next step is to derive the \emph{implementation-level compositional semantics} (ICS) by choosing a CRDT decomposition and defining transformation rules (\Cref{def:translate}). A natural decomposition mirrors a sequential design: a replicated node set (\nodeset) and a map from each node to its outgoing-edge set (\edgemap). The transformation rules (Fig.~\ref{fig:layered_translation_id}) map each application event to one or more component updates. For nodes, \addN{\texttt{n}} and \rmvN{\texttt{n}} translate to \add{\texttt{n}} and \del{\texttt{n}}; to realize isolate-delete under concurrency, \addE{\texttt{n}}{\texttt{m}} additionally produces \add{\texttt{n}} and \add{\texttt{m}} on the \nodeset, inducing the conflict needed to win over a concurrent node removal. For edges, outgoing edges are stored in an adjacency map: \addE{\texttt{n}}{\texttt{m}} translates to \upd{\texttt{n}}{\add{\texttt{m}}} and \rmvE{\texttt{n}}{\texttt{m}} to \upd{\texttt{n}}{\del{\texttt{m}}}; we use \texttt{upd} rather than key removal, which would delete all outgoing edges of~\texttt{n}.
Since \edgemap\, is a Map$\langle$Set$\rangle$, its internal transformation rules also apply; Map$\langle$Set$\rangle$ is fully defined in the extended version.
\begin{definition} [Transformation rules]
A \emph{transformation rule} $\mathcal{T}: Op_{X} \to \mathcal{P}(Op_{X_1} \cup \dots \cup Op_{X_n})$ maps each operation of a composite data type $X$ to a set of operations over its components $X_1, \dots, X_n$.
\label{def:translate}
\end{definition}
\vspace{-0.8em}
\begin{figure}[h!]
\centering
{\footnotesize
\[
\renewcommand{\arraystretch}{1.15}
\begin{array}{rcl}
\idg.\addN{n}
&\overset{\mathcal{T}}\Rightarrow&
\{ \nodeset.\add{n} \}\\

\idg.\rmvN{n}
&\overset{\mathcal{T}}\Rightarrow&
\{ \nodeset.\del{n} \}\\

\idg.\addE{n}{m}
&\overset{\mathcal{T}}\Rightarrow& \{\nodeset.\add{n},\nodeset.\add{m}, \edgemap.\upd{n}{\add{m}}\}\\

\idg.\rmvE{n}{m}
&\overset{\mathcal{T}}\Rightarrow& 
\{ \edgemap.\upd{n}{\del{m}} \}\\

\edgemap.\upd{n}{\add{m}} 
&\overset{\mathcal{T}}\Rightarrow&
\{\keyset.\add{n},\ \valset_n.\add{m}\} \\

\edgemap.\upd{n}{\del{m}}
&\overset{\mathcal{T}}\Rightarrow&
\{\keyset.\add{n}, \valset_n.\del{m}\}
\end{array}
\]
}
\caption{Layered transformation for isolate-delete graph.}
\label{fig:layered_translation_id}
\end{figure}
The resulting ICS of the ID graph is given as:
\begin{equation}    
{\footnotesize
\begin{aligned}
&\F{IDGraph\_ICS}\big(\hasN{n}, (\E, \op, \vis, \ar)\big) = \text{true} \iff \F{Set}\big(\hasV{n}, (\E_{\nodeset}, \op_{\nodeset}, \vis, \ar)\big) \\
&\F{IDGraph\_ICS}\big(\hasE{n}{m}, (\E, \op, \vis, \ar)\big) = \text{true}  \iff \\ &\quad \F{MapSet}\big(\hasKV{n}{m}, (\E_{\edgemap}, \op_{\edgemap}, \vis, \ar)\big) 
\end{aligned}
}
\label{eq:graphidicsnode}
\end{equation}
In the Datalog realization (Listing \ref{lst:graphidics}), the input facts have the same signature as the SLS, namely \dtinline{graphEvent(E, O, V1, V2)}. 
The resulting state is computed for each component. 
Note that this implementation may produce keys with empty value sets. 
We interpret such entries as non-existent edges and ignore them when inspecting the application state.

\begin{lstlisting}[style=souffleStyle, caption = ICS for isolate-delete graph., label = lst:graphidics]
nodesEvent(E, "add", V) :- graphEvent(E, "addN", V, _).
nodesEvent(E, "del", V) :- graphEvent(E, "rmvN", V, _).
nodesEvent(E, "add", V) :- graphEvent(E, "addE", V, _).
nodesEvent(E, "add", V) :- graphEvent(E, "addE", _, V).
edgesMapSetEvent(E, "upd", "add", V1, V2) :- graphEvent(E, "addE", V1, V2).
edgesMapSetEvent(E, "upd", "del", V1, V2) :- graphEvent(E, "rmvE", V1, V2).
\end{lstlisting}

The framework is compositional at the specification level: composite CRDTs are written as decompositions of simpler ones together with the corresponding transformation rules. The ICS is then mechanically derived from this choice of decomposition. However, Soufflé does not allow multiple independent instances of the same CRDT to share a single set of rules without an explicit discriminator. For nested types, this can be handled through a shape-based encoding, where nesting positions are represented as structured scopes. In this approach, a single set module serves all (map) nesting levels, with each instance distinguished by its position in a recursive scope type rather than by duplicated and renamed Datalog rules. For example, a Map$\langle$Map$\langle$Set$\rangle \rangle$ reuses the same set semantics as a Map$\langle$Set$\rangle$, with an additional scope level encoding the outer key: no ad-hoc adjustments are required, and the pattern extends to arbitrary nesting depths. In this paper, we opt for the explicit, per-instance encoding (e.g., separate \texttt{nodesEvent} and \texttt{edgesMapSetEvent} relations) as it maps directly to the formal definitions in Section~\ref{background} and makes the correspondence between the specification-level and implementation-level semantics transparent for verification.

CRDT composition is not semantics-neutral: decomposition choices determine which operations conflict and how conflicts are resolved. For graphs, nodes and edges are separate structures but semantically coupled; i.e., an \texttt{addE} may defeat a concurrent \texttt{rmvN}, so the transformation rules must intentionally induce or avoid such conflicts.

%% file: workflow_diagram_v2.tex
  \usetikzlibrary{shapes.geometric, positioning, fit}                                                                                                             
  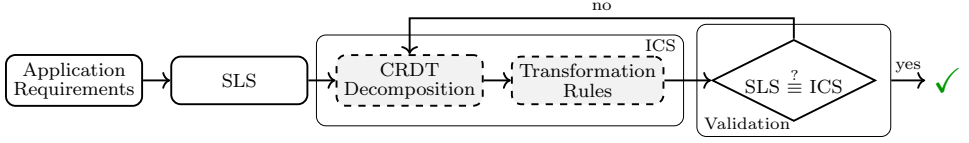
\begin{figure}[t]                                                                                                                                               
  \centering                                                                                                                                                      
  \scalebox{.9}{                                                                                                                                                  
  \begin{tikzpicture}[
    font=\footnotesize,
    node distance=4mm and 4mm,
    box/.style={                                                                                                                                                  
      draw, rounded corners, thick, align=center,
      minimum width=2cm, minimum height=0.7cm                                                                                                                     
    },            
    abstract/.style={box},
    process/.style={box, fill=gray!10, dashed},
    decision/.style={draw,
    diamond,
    thick,
    align=center,
    aspect=2,
    inner sep=1pt},                                                                                                                        
    arrow/.style={->, thick},
    bigbox/.style={draw, rounded corners, inner sep=5pt}                                                                                                          
  ]               

  \node[box] (app) {Application\\[-1pt]Requirements};
  \node[abstract, right=of app] (adt) {SLS};                                                                                                                      
  \node[process, right=of adt] (decomp) {CRDT\\[-1pt]Decomposition};
  \node[process, right=of decomp] (trans) {Transformation\\[-1pt]Rules};                                                                                          
   
  \node[bigbox, fit=(decomp)(trans), inner sep=8pt] (icsbox) {};                                                                                                  
  \node[font=\scriptsize, anchor=north east] at ([yshift=1.3pt]icsbox.north east) {ICS};
                                                                                                                                                                  
  \node[decision, right=of icsbox] (decide) {SLS $\stackrel{?}{\equiv}$ ICS};                                                                                                       
   
  \node[bigbox, fit=(decide), inner sep=6pt] (valbox) {};                                                                                                         
  \node[font=\scriptsize, anchor=south west] at ([yshift=0.1pt]valbox.south west) {Validation};

  \node[right=5mm of valbox, text=green!60!black] (ok) {\Large\checkmark};
                                                                                                                                                                  
  \draw[arrow] (app) -- (adt);
  \draw[arrow] (adt) -- (decomp);
  \draw[arrow] (decomp) -- (trans);                                                                                                                               
  \draw[arrow] (trans) -- (decide);
  \draw[arrow] (valbox) -- node[above]{\scriptsize yes} (ok);                                                                                                     
                                                                                                                                                                  
  \draw[arrow] (decide.north) -- ++(0,3mm) -| node[near start, above]{\scriptsize no} (decomp.north);                                                             
                                                                                                                                                                  
  \end{tikzpicture}}
  \vspace{2mm}
  \caption{Workflow for deriving and validating CRDT compositions. From application requirements, the user defines the SLS, decomposes it into CRDTs    
  with transformation rules (ICS), and validates equivalence. Mismatches require revising the decomposition.}                                                
  \label{fig:workflow}
  \end{figure}

%% file: use-case.tex
We evaluate \textsc{CRDTlog} on two directed graph variants: the \emph{isolate-delete (ID)} graph from \Cref{background,spec} and a \emph{detach-delete (DD)} graph introduced below, capturing common deletion policies with non-trivial concurrent interactions. We show that \textsc{CRDTlog} yields precise, executable semantics and enables systematic CRDT component selection.

\paragraph{Isolate-delete semantics (ID).} 
We first revisit the ID semantics defined in Section \ref{souffleCRDTApplications}. A natural implementation-level decomposition (ICS) represents nodes as an add-wins Set and outgoing edges as a $\mathrm{Map} \langle \mathrm{Set} \rangle$. The transformation layer is straightforward for most operations and maps \texttt{addE} to additional \texttt{add} updates on the node set to ensure that concurrent edge additions override node removals, as required by the SLS.

\paragraph{Detach-delete semantics (DD).}

We next consider the DD semantics for directed graphs. This specifies that removing a node also removes all its incident edges, ensuring the absence of dangling edges.
The formal definition differs from ID only in the definition of $\hasE{u}{v}$ (\Cref{tab:ID-graph}).
DD preserves the node semantics of ID. In particular, a node $n$ is present if it has been added and not subsequently removed, unless the removal is invalidated by a concurrent edge addition involving $n$. Hence, the definition of $\F{DDGraph\_SLS}(\hasN{n},\dots)$ is identical to the isolate-delete case (Eq. \ref{eq:graphADT}).
Further, a DD graph contains an edge \texttt{(n,m)} if it has been added, has not been subsequently removed, and neither the source nor target node has been removed. 
The node condition is the new part: if \dtinline{n} (or \dtinline{m}) is removed, then all incident edges are removed as well (Eq. \ref{eq:graphddADTEdges}, Listing \ref{lst:graphddadtedges}). In the listing, \dtinline{edges(N,M)} holds for added edges whose endpoints were not removed, while \dtinline{postVisRmvN} captures concurrent \dtinline{addE} that invalidate node removals.
\begin{equation}
\label{eq:graphddADTEdges}
{\footnotesize
\begin{aligned}[t]
&\F{DDGraph\_SLS}\big(\hasE{n}{m},\E,\op,\vis,\ar\big)=\text{true}\\
&\iff \exists e\in \E:\ \addE{n}{m} \in \op(e)\ \land\\
&\quad \quad  \neg\exists e'\in E:\ \rmvE{n}{m} \in \op(e') \land e\visrel e'\ \land
\big(\Gamma(e,n)\land \Gamma(e,m)\big) \\ 
&\Gamma(e,u)\stackrel{\mathrm{def}}{=}
\neg\exists s\in \E:\ \rmvN{u} \in \op(s)\land e\visrel s\ 
\end{aligned}
}
\end{equation}
%

\begin{lstlisting}[style=souffleStyle, caption = SLS for DD graph (edge rules)., label = lst:graphddadtedges]
postVisRmvN(E, K, S) :- graphEvent(E, _, _, _), graphEvent(S, "rmvN", K, _), 
                              visible(E, S), E \= S.
postVisRmvE(E)       :- graphEvent(E, _, K1, K2), graphEvent(S, "rmvE", K1, K2), 
                              visible(E, S), E \= S.
edges(N, M) :- graphEvent(E, "addE", N, M), not postVisRmvE(E), 
               not postVisRmvN(E, N, _), not postVisRmvN(E, M, _).
\end{lstlisting}


At first glance, one might attempt to reuse the ID implementation strategy, using a set for nodes and a Map$\langle$Set$\rangle$ for edges. Thus, removing a node \dtinline{n} would remove \dtinline{n} from the node set and delete its outgoing edges by removing the key \dtinline{n} from the edge map. Incoming edges could then be removed by deleting the corresponding entries from the value sets associated with other keys.
However, this straightforward approach fails because it introduces unintended conflicts in the map. In particular, a concurrent \rmvE{n}{m} and \rmvN{n} generate conflicting updates that combine a value-level update (a \texttt{upd} or inner-set modification) with a key-level \texttt{rmv}. 
Under the update-wins map semantics, the \texttt{upd} may resurrect the key \dtinline{n}, thereby preserving outgoing edges from \dtinline{n}. 
The resulting state contains dangling edges, which violates the DD semantics.
To avoid this, we represent the graph with one set for nodes and one set for edges. The edge set stores tuples \dtinline{(n,m)}. This design eliminates key-level put-versus-remove conflicts since all edge updates are applied uniformly to a single set.
The transformations for node operations are unchanged from the ID case. The transformation for edge operations is defined in Fig. \ref{fig:transformation_dd}.

\begin{figure}[h!]
\centering
{\footnotesize
\[
\renewcommand{\arraystretch}{1.15}
\begin{array}{rcl}
\ddg.\rmvN{n}
&\overset{\mathcal{T}}\Rightarrow&
\{\edgeset.\del{n,m}, \edgeset.\del{m,n} \}, \\ &&\text{ where } m \text{ is a node to/from which an edge from/to } n \text{ exists.}\\

\ddg.\addE{n}{m}
&\overset{\mathcal{T}}\Rightarrow& \{ \edgeset.\add{n,m} \}\\ 

\ddg.\rmvE{n}{m}
&\overset{\mathcal{T}}\Rightarrow& \{ \edgeset.\del{n,m} \}
\end{array}
\]
}
\caption{Transformation rules for detach-delete graph. }
\label{fig:transformation_dd}
\end{figure}

\begin{equation}
\footnotesize
\begin{aligned}
& \F{DDGraph\_ICS}\big(\hasE{n}{m}, (\E, \op, \vis, \ar)\big) = \text{true}  \\
& \iff  \F{Set}\big(\hasV{n,m}, (\E_{\edgeset}, \op_{\edgeset}, \vis, \ar)\big) 
\end{aligned}
\label{eq:graphddicsedge}
\end{equation}

\begin{lstlisting}[style=souffleStyle, caption = ICS for DD graph. Transformation rules for the edges., label = lst:graphddicsedges]
edgesEvent(E, "add", V1, V2) :- graphEvent(E, "addE", V1, V2).
edgesEvent(E, "del", V1, V2) :- graphEvent(E, "rmvE", V1, V2).
edgesEvent(E, "del", V1, V2) :- graphEvent(E, "rmvN", V1, _), preVisAddOutE(E, V2).
edgesEvent(E, "del", V2, V1) :- graphEvent(E, "rmvN", V1, _), preVisAddInE(E, V2).
preVisAddInE(E, V2)  :- graphEvent(E, _, V1, _), graphEvent(S, "addE", V2, V1),
                        visible(S, E), E\=S.
preVisAddOutE(E, V2) :- graphEvent(E, _, V1, _), graphEvent(S, "addE", V1, V2),
                        visible(S, E), E\=S.
\end{lstlisting}

Note that the two additional predicates \dtinline{preVisAddInE} and \dtinline{preVisAddOutE} identify the edges that are visible during the removal event. The transformation then removes exactly these edges, ensuring that successful node deletions detach all incident edges while preserving the intended concurrency behavior captured by the node semantics.

We establish that SLS and ICS in the DD graph preserve the no-dangling edge invariant. The full proofs are in the supplemental material.
\begin{lemma}[Dangling Edge Safety for the SLS of a DD Graph]\label{lem:dd-sls-nodangle} 
For the SLS of the DD graph, the invariant of no dangling edges holds for any valid operation context.
\end{lemma}
\begin{proof}[Proof sketch]
A dangling edge requires a removed node $n$ with a present incident edge. If \rmvN{$n$} has no concurrent incident \texttt{addE}, the endpoint guard $\Gamma(e,n)$ (\Cref{eq:graphddADTEdges}) is falsified by the observed removal for every incident \texttt{addE} event $e$, so no such edge survives. If a concurrent \addE{$n$}{$m$} exists, $\Gamma(e,n)$ holds as the removal is not visible to $e$, but the node predicate (\Cref{eq:graphADT}) also retains $n$, so the edge is not dangling.
\end{proof} 
\begin{lemma}[Dangling Edge Safety for the ICS of a DD Graph]\label{lem:dd-ics-nodangle}
For the ICS of the DD graph, the invariant of no dangling edges holds for any valid operation context.
\end{lemma}
\begin{proof}[Proof sketch]
As above, we case-split on events concurrent with \rmvN{$n$}. Without a concurrent \texttt{addE}, the operation translates into \texttt{del}s in the underlying set (\Cref{fig:transformation_dd}), removing all incident edges with $n$. With a concurrent \addE{$n$}{$m$}, the transformed \add{$n$} wins over \texttt{del} in the nodes set, retaining $n$ and preventing a dangling edge.
\end{proof}

This case study highlights the value of explicit semantic specifications: they expose situations in which a CRDT decomposition enforces conflict resolution incompatible with the intended application semantics. \textsc{CRDTlog} makes this mismatch explicit at the specification level, supporting designers in composing CRDTs with the desired semantics.

%% file: experiments.tex
We compare ICS and SLS for both ID and DD graph applications. Correctness is assessed via property-based testing: for eight configurations, we generate 1K independent executions and verify that SLS and ICS produce identical output graphs. We then benchmark scalability by varying event count, replica count, and graph size. We structure the analysis around the following research questions.

\begin{itemize}
\item \emph{\textbf{RQ1} (Semantic Equivalence).}
Do ICS and SLS yield the same observable states for all tested operation contexts?

\item \emph{\textbf{RQ2} (Application Graph Scalability).}
How does the size of the application graph affect the cost of semantic evaluation?

\item \emph{\textbf{RQ3} (Concurrency Scalability).}
How does increasing replica count and concurrency affect evaluation cost?

\item \emph{\textbf{RQ4} (Event Scalability).}
How does evaluation cost scale with increasing events?

\item \emph{\textbf{RQ5} (Incrementality).}
How does incremental evaluation impact verification?
\end{itemize}

\paragraph{Experimental Setup.}

All experimental inputs were automatically generated using Python scripts that enforce  each graph application's preconditions. For example, in the isolate-delete variant,  \dtinline{rmvN(n)} is never generated for nodes with incident edges,
ensuring that all executions satisfy the assumptions of the model.
The primary goal is to validate semantic equivalence (SLS $\equiv$ ICS) via property-based testing, therefore the inputs must be large and diverse enough to cover corner cases arising from concurrent operations across multiple replicas. While the graph semantics used here are intentionally minimal, the framework is intended for richer application domains (e.g., property graphs, collaborative spreadsheets) where complex operation interleavings are expected. For such scenarios, generating large workloads is necessary to provide adequate coverage.

Equivalence tests ran on a machine running  Ubuntu 22.04.2 with 11th Gen Intel(R) Core(TM) i7-1165G7 2.80 GHz CPU and 32 GB of memory. 

Scalability experiments ran on a CPU-only x86\_64 server with two AMD EPYC 7443 processors (1.5--2.85 GHz, boost enabled) and 100 GB of allocated memory. For each configuration, we report mean execution times over 30 independent runs. Both SLS and ICS were evaluated on the same generated inputs.

\paragraph{Semantic Equivalence (\textbf{RQ1}).}
SLS and ICS compute identical observable states for all tested operation contexts. We evaluated semantic equivalence using a property-based testing approach \citep{DBLP:conf/icfp/ClaessenH00}, generating randomized operation contexts and comparing the outputs of SLS and ICS. We considered eight configurations, combining two replica counts (5 and 10) with four event volumes (20, 50, 100, and 1K events). For each configuration, we generated 1,000 independent test inputs. All executions terminated successfully, and in every case, the SLS and ICS implementations produced identical results. These results provide empirical evidence that our CRDT compositions faithfully realize the intended semantics for both the ID and DD graph variants. Average execution times per test are reported in the supplemental material.

\input{figures/fig_scalability}

\paragraph{Scalability (\textbf{RQ2}--\textbf{RQ5}).}
We analyze how semantic evaluation cost scales with operation context size, concurrency, and application graph size. Figure~\ref{fig:scalability_all} reports runtime as a function of graph size (Fig.~\ref{fig:scalability_all}a), replica count (Fig.~\ref{fig:scalability_all}b), and number of events (Fig.~\ref{fig:scalability_all}c). For concurrency and event scalability, we used the full range of graph operations, whereas for the graph scalability we used only \dtinline{addN} and \dtinline{addE} events. Concurrency is modeled by allowing operations to interleave and branch, producing non-trivial visibility structures.

\paragraph{Application Graph Scalability (\textbf{RQ2}).}
Evaluation cost increases with application graph size (nodes and edges), but ICS scales better than SLS for larger graphs. To study graph size scalability, we sampled random graphs with a fixed number of nodes (10, 100, and 1K) using an edge probability of $0.05$, yielding an expected number of edges proportional to the graph size ($\approx$ 15, 600, and 50K). This reflects the average number of events. For each graph, we generated an operation context with 4 replicas. 

\Cref{fig:scale_graph_sub} shows that for small graphs, ICS implementations are slightly slower than the corresponding SLS specifications due to the overhead of transformation rules. However, as graph size increases (and with it the number of events in the operation context), ICS becomes consistently faster for both applications. At approximately 50K graph entities, the ICS realization is nearly twice as fast as SLS for ID and close to an order of magnitude faster for DD. This behavior reflects the lower rule complexity of the component CRDT semantics used in ICS, which outweighs the transformation overhead once evaluation is dominated by joins over many events.

\paragraph{Concurrency Scalability (\textbf{RQ3}).}
Increasing the number of replicas only mildly impacts the evaluation cost. We varied the replica count across 2, 4, 8, and 16, holding the total event count fixed at 1K and 10K. As shown in \Cref{fig:scale_repl_sub}, increasing the replica count results in only modest runtime growth. In \textsc{CRDTlog}, the number of replicas mainly affects the branching structure of the visibility relation without substantially increasing the complexity of rule evaluation. This suggests that semantic evaluation scales well with respect to concurrency.

\paragraph{Event Scalability (\textbf{RQ4}).}

The evaluation cost grows primarily with the number of events in the operation context. As summarized in \Cref{fig:scalability_all}c, runtime increases steadily as the number of events grows for both graph variants and both semantic encodings. This trend holds largely independently of the number of replicas, indicating that event volume is the dominant factor driving semantic evaluation cost. Unlike in RQ2, the resulting graph size is not held constant here: since add and remove operations are interleaved during generation, graph size varies across runs and is not an independent variable.

\paragraph{Incrementality (\textbf{RQ5}).} 
Our framework can also be instantiated with DDlog, an incremental Datalog engine that supports an interactive workflow: users incrementally add events and observe the resulting CRDT state, aiding scenario exploration and debugging. We evaluate DDlog in single-shot and incremental mode across varying events, replicas, and graph sizes. In incremental mode, events are inserted one at a time, mirroring the interactive use case. Both modes are substantially slower than Souffl\'{e} (\Cref{fig:sinc_ops_sub}), which is expected. DDlog optimizes for small updates over a stable base, so it pays a constant per-tuple cost for incrementality, even in single-shot mode. Here each new event causes changes to the visibility relation, so incremental evaluation effectively recomputes everything. DDlog
remains practical for small to medium contexts (up to $\approx$ 1K events), which suffices for its intended interactive, exploratory use case.

\input{figures/fig_ddlog_scalability}

\paragraph{Discussion.}

Although we do not directly compare ID and DD semantics, their scaling behaviors are notably distinct. DD is substantially faster for large graphs, with the SLS–ICS gap most pronounced. This reflects the underlying representations: ID stores edges in a Map$\langle$Set$\rangle$, introducing nesting and multiple component-level updates per event, whereas DD uses a single set of edge tuples with at most one component update per event. ICS consistently outperforms SLS for both graphs, as decomposition yields smaller relations with localized joins rather than complex guard predicates over the entire event. The effect is amplified in DDlog, where the incremental engine benefits from change propagation over many small relations. Thus, decomposition can significantly reduce intermediate joins and improve evaluation time, even when the resulting semantics are identical.

While runtimes grow steeply beyond 1K events, we argue that this is sufficient for property-based testing. Optimizing the underlying Datalog engine is orthogonal to our contribution and could further extend the feasible testing range.

%% file: figures/fig_scalability.tex
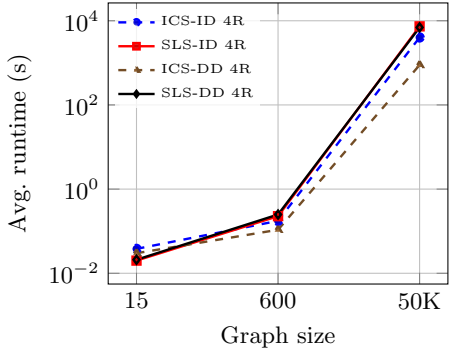
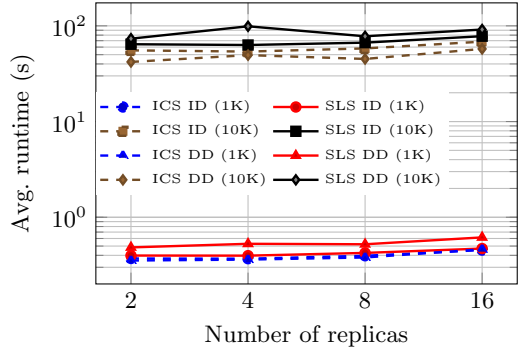
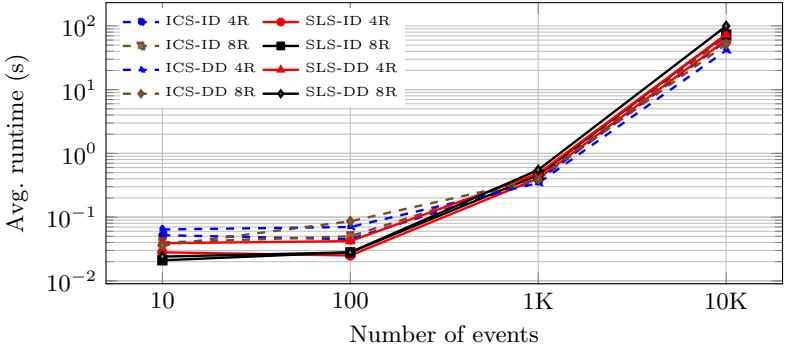
\begin{figure}
\centering

\begin{subfigure}[t]{0.45\linewidth}
\centering
\begin{tikzpicture}
\begin{axis}[
    width=\linewidth,
    height=5.3cm,
    xlabel={Graph size},
    ylabel={Avg. runtime (s)},
    xmode=log,
    log basis x=10,
    xtick={10,100,1000},
    xticklabels={15,600,50K},
    ymode=log,
    yticklabel style={
        /pgf/number format/fixed,
        /pgf/number format/precision=3
    },
    grid=both,
    legend columns=1,
    legend style={
        font=\tiny,
        at={(0.02,0.98)},
        anchor=north west,
        draw=none,
        fill=white,
        fill opacity=0.80,
        text opacity=1,
        cells={anchor=west},
        inner sep=0.6pt,
        row sep=0.2pt,
    },
    legend image post style={scale=0.8},
    every axis plot/.append style={mark size=1.5pt, line width=0.9pt},
]
\addplot+[dashed, mark=*]          coordinates {(10,0.038) (100,0.176) (1000,4014.265)};
\addlegendentry{ICS-ID 4R}
\addplot+[solid, mark=square*]    coordinates {(10,0.020) (100,0.227) (1000,7360.106)};
\addlegendentry{SLS-ID 4R}
\addplot+[dashed, mark=triangle*] coordinates {(10,0.030) (100,0.108) (1000,888.227)};
\addlegendentry{ICS-DD 4R}
\addplot+[solid, mark=diamond]  coordinates {(10,0.021) (100,0.249) (1000,6940.607)};
\addlegendentry{SLS-DD 4R}
\end{axis}
\end{tikzpicture}
\vspace{-2mm}
\caption{Scaling with graph size. The number of events is equivalent to the graph size.} 
\label{fig:scale_graph_sub}
\end{subfigure}%
\hfill
\begin{subfigure}[t]{0.53\linewidth}
\centering
\begin{tikzpicture}
\begin{axis}[
    width=\linewidth,
    height=5.3cm,
    xlabel={Number of replicas},
    ylabel={Avg. runtime (s)},
    xmode=log,
    log basis x=2,
    xtick={2,4,8,16},
    xticklabels={2,4,8,16},
    ymode=log,
    yticklabel style={
        /pgf/number format/fixed,
        /pgf/number format/precision=3
    },
    grid=both,
    legend columns=2,
    legend style={
        font=\tiny,
        at={(0.00,0.50)},
        anchor=west,
        draw=none,
        fill=white,
        fill opacity=1,
        text opacity=1,
        cells={anchor=west},
        inner sep=1pt
    },
    every axis plot/.append style={mark size=1.7pt, line width=0.9pt},
]

\addplot+[dashed,  mark=*]        coordinates {(2,0.371) (4,0.369) (8,0.396) (16,0.454)};
\addlegendentry{ICS ID (1K)}
\addplot+[solid,  mark=*]  coordinates {(2,0.398) (4,0.397) (8,0.424) (16,0.471)};
\addlegendentry{SLS ID (1K)}
\addplot+[dashed, mark=square*]        coordinates {(2,55.438) (4,53.857) (8,58.123) (16,68.398)};
\addlegendentry{ICS ID (10K)}
\addplot+[solid, mark=square*]  coordinates {(2,64.279) (4,63.093) (8,66.780) (16,78.013)};
\addlegendentry{SLS ID (10K)}

\addplot+[dashed,  mark=triangle*] coordinates {(2,0.353) (4,0.361) (8,0.380) (16,0.462)};
\addlegendentry{ICS DD (1K)}
\addplot+[solid,  mark=triangle*]  coordinates {(2,0.484) (4,0.528) (8,0.522) (16,0.616)};
\addlegendentry{SLS DD (1K)}
\addplot+[dashed, mark=diamond*] coordinates {(2,41.990) (4,49.455) (8,45.117) (16,57.427)};
\addlegendentry{ICS DD (10K)}
\addplot+[solid, mark=diamond*]  coordinates {(2,73.488) (4,98.949) (8,77.828) (16,91.526)};
\addlegendentry{SLS DD (10K)}

\end{axis}
\end{tikzpicture}
\vspace{-2mm}
\caption{Scaling with the number of replicas for 1K and 10K events.}
\label{fig:scale_repl_sub}
\end{subfigure}

\vspace{0.6em}

\begin{subfigure}[t]{0.78\linewidth}
\centering
\begin{tikzpicture}
\begin{axis}[
    width=\linewidth,
    height=5.3cm,
    xlabel={Number of events},
    ylabel={Avg. runtime (s)},
    xmode=log,
    log basis x=10,
    xtick={10,100,1000,10000},
    xticklabels={10,100,1K,10K},
    ymode=log,
    yticklabel style={
        /pgf/number format/fixed,
        /pgf/number format/precision=3
    },
    grid=both,
    legend columns=2,
    legend style={
        font=\tiny,
        at={(0.02,0.98)},
        anchor=north west,
        draw=none,
        fill=white,
        fill opacity=1,
        text opacity=1,
        cells={anchor=west},
        inner sep=0.6pt,
        row sep=0.2pt,
    },
    legend image post style={scale=0.8},
    every axis plot/.append style={mark size=1.5pt, line width=0.9pt},
]

\addplot+[dashed, mark=*]         coordinates {(10,0.052) (100,0.045) (1000,0.377) (10000,58.281)};
\addlegendentry{ICS-ID 4R}
\addplot+[solid, mark=*]   coordinates {(10,0.028) (100,0.025) (1000,0.406) (10000,60.097)};
\addlegendentry{SLS-ID 4R}
\addplot+[dashed, mark=square*]        coordinates {(10,0.040) (100,0.050) (1000,0.429) (10000,67.333)};
\addlegendentry{ICS-ID 8R}
\addplot+[solid, mark=square*]  coordinates {(10,0.021) (100,0.028) (1000,0.457) (10000,73.376)};
\addlegendentry{SLS-ID 8R}

\addplot+[dashed, mark=triangle*]  coordinates {(10,0.064) (100,0.070) (1000,0.342) (10000,41.254)};
\addlegendentry{ICS-DD 4R}
\addplot+[solid, mark=triangle*]   coordinates {(10,0.039) (100,0.042) (1000,0.483) (10000,73.263)};
\addlegendentry{SLS-DD 4R}
\addplot+[dashed, mark=diamond*] coordinates {(10,0.036) (100,0.086) (1000,0.393) (10000,51.544)};
\addlegendentry{ICS-DD 8R}
\addplot+[solid, mark=diamond*]  coordinates {(10,0.024) (100,0.028) (1000,0.551) (10000,99.462)};
\addlegendentry{SLS-DD 8R}

\end{axis}
\end{tikzpicture}
\vspace{-2mm}
\caption{Scaling with the number of events. As the add and remove operations randomly interleave, the graph size varies across the runs.}
\label{fig:scale_ops_sub}
\end{subfigure}
\vspace{1mm}
\caption{Scalability of SLS and ICS for isolate-delete (ID) and detach-delete (DD).
}
\label{fig:scalability_all}
\end{figure}

%% file: figures/fig_ddlog_scalability.tex
  \begin{figure}[t]                                                             
  \centering
                                                                                
  \begin{subfigure}[t]{0.45\linewidth}
  \centering                                                                    
  \begin{tikzpicture}
  \begin{axis}[                                                                 
      width=\linewidth,
      height=5.3cm,
      xlabel={Graph size},
      ylabel={Avg. runtime (s)},                                                
      xmode=log,
      log basis x=10,                                                           
      xtick={10,100,1000},
      xticklabels={15,600,1K},
      ymode=log,                                                                
      yticklabel style={
          /pgf/number format/fixed,                                             
          /pgf/number format/precision=3
      },
      grid=both,
      legend columns=1,
      legend style={                                                            
          font=\tiny,
          at={(0.02,0.98)},                                                     
          anchor=north west,
          draw=none,
          fill=white,
          fill opacity=0.80,
          text opacity=1,
          cells={anchor=west},                                                  
          inner sep=0.6pt,
          row sep=0.2pt,                                                        
      },          
      legend image post style={scale=0.8},
      every axis plot/.append style={mark size=1.5pt, line width=0.9pt},
  ]                                                                             
  
  \addlegendimage{solid, black, mark=*, mark size=1.5pt, line width=0.9pt}
  \addlegendentry{Single}                                                       
  \addlegendimage{dashed, black, mark=*, mark options={fill=white}, mark
  size=1.5pt, line width=0.9pt}                                                 
  \addlegendentry{Incremental}
                                                                                
  \addplot[solid, blue, mark=*, mark size=1.5pt, line width=0.9pt]
    coordinates {(10,0.052)(100,1.380)};                                        
  \addlegendentry{ICS-ID}                                                       
  \addplot[solid, red, mark=square*, mark size=1.5pt, line width=0.9pt]         
    coordinates {(10,0.056)(100,2.646)};                                        
  \addlegendentry{SLS-ID}
  \addplot[solid, brown!60!black, mark=triangle*, mark size=1.5pt, line         
  width=0.9pt]                                                                  
    coordinates {(10,0.055)(100,0.553)};
  \addlegendentry{ICS-DD}                                                       
  \addplot[solid, orange, mark=diamond*, mark size=1.5pt, line width=0.9pt]     
    coordinates {(10,0.060)(100,3.524)};
  \addlegendentry{SLS-DD}                                                       
                  
  \addplot[dashed, blue, mark=*, mark options={fill=white}, mark size=1.5pt,
  line width=0.9pt, forget plot]                                                
    coordinates {(10,0.081)(100,6.411)};
  \addplot[dashed, red, mark=square*, mark options={fill=white}, mark           
  size=1.5pt, line width=0.9pt, forget plot]                                    
    coordinates {(10,0.073)(100,6.234)};
  \addplot[dashed, brown!60!black, mark=triangle*, mark options={fill=white},   
  mark size=1.5pt, line width=0.9pt, forget plot]                               
    coordinates {(10,0.082)(100,4.347)};
  \addplot[dashed, orange, mark=diamond*, mark options={fill=white}, mark       
  size=1.5pt, line width=0.9pt, forget plot]
    coordinates {(10,0.081)(100,10.644)};                                       
                                                                                
  \end{axis}
  \end{tikzpicture}                                                             
  \vspace{-2mm}   
  \caption{Scaling with graph size}
  \label{fig:sinc_graph_sub}
  \end{subfigure}%
  \hfill
  \begin{subfigure}[t]{0.53\linewidth}
  \centering
  \begin{tikzpicture}                                                           
  \begin{axis}[
      width=\linewidth,                                                         
      height=5.3cm,
      xlabel={Number of replicas},
      ylabel={Avg. runtime (s)},
      xmode=log,
      log basis x=2,                                                            
      xtick={2,4,8,16},
      xticklabels={2,4,8,16},                                                   
      ymode=log,  
      yticklabel style={
          /pgf/number format/fixed,
          /pgf/number format/precision=3                                        
      },
      grid=both,                                                                
      legend columns=3,
      legend style={
          font=\tiny,
          at={(0.00,0.45)},
          anchor=west,
          draw=none,
          fill=white,                                                           
          fill opacity=1,
          text opacity=1,                                                       
          cells={anchor=west},
          inner sep=1pt,
      },
      legend image post style={scale=0.8},
      every axis plot/.append style={mark size=1.5pt, line width=0.9pt},        
  ]
                                                                                
  \addlegendimage{solid, black, mark=*, mark size=1.5pt, line width=0.9pt}
  \addlegendentry{Single}                                                       
  \addlegendimage{dashed, black, mark=*, mark options={fill=white}, mark
  size=1.5pt, line width=0.9pt}                                                 
  \addlegendentry{Incremental}
                                                                                
  \addplot[solid, blue, mark=*, mark size=1.5pt, line width=0.9pt]              
    coordinates {(2,2.183)(4,2.535)(8,2.652)(16,3.553)};
  \addlegendentry{ICS-ID 1K}                                                    
  \addplot[solid, red, mark=square*, mark size=1.5pt, line width=0.9pt]
    coordinates {(2,3.970)(4,2.699)(8,4.671)(16,3.427)};                        
  \addlegendentry{SLS-ID 1K}                                                    
  \addplot[solid, green!60!black, mark=diamond*, mark size=1.5pt, line          
  width=0.9pt]                                                                  
    coordinates {(2,1.774)(4,1.737)(8,1.658)(16,2.150)};
  \addlegendentry{ICS-DD 1K}                                                    
  \addplot[solid, orange, mark=pentagon*, mark size=1.5pt, line width=0.9pt]
    coordinates {(2,415.264)(4,397.948)(8,466.507)(16,480.736)};                
  \addlegendentry{SLS-DD 1K}                                                    
  \addplot[solid, cyan!60!black, mark=*, mark size=1.5pt, line width=0.9pt]     
    coordinates {(2,388.162)(4,488.751)(8,619.419)(16,755.692)};                
  \addlegendentry{ICS-ID 10K}
  \addplot[solid, brown!60!black, mark=triangle*, mark size=1.5pt, line         
  width=0.9pt]    
    coordinates {(2,711.610)(4,512.340)(8,882.261)(16,745.740)};                
  \addlegendentry{SLS-ID 10K}                                                   
  \addplot[solid, violet, mark=diamond*, mark size=1.5pt, line width=0.9pt]
    coordinates {(2,300.575)(4,321.112)(8,312.485)(16,441.953)};                
  \addlegendentry{ICS-DD 10K}
                                                                                
  \addplot[dashed, blue, mark=*, mark options={fill=white}, mark size=1.5pt,
  line width=0.9pt, forget plot]                                                
    coordinates {(2,11.119)(4,6.970)(8,6.606)(16,6.625)};
  \addplot[dashed, red, mark=square*, mark options={fill=white}, mark           
  size=1.5pt, line width=0.9pt, forget plot]
    coordinates {(2,15.750)(4,13.477)(8,9.675)(16,13.391)};                     
  \addplot[dashed, green!60!black, mark=diamond*, mark options={fill=white},    
  mark size=1.5pt, line width=0.9pt, forget plot]                               
    coordinates {(2,6.522)(4,5.621)(8,5.344)(16,5.357)};                        
  \addplot[dashed, orange, mark=pentagon*, mark options={fill=white}, mark      
  size=1.5pt, line width=0.9pt, forget plot]                                    
    coordinates {(2,1063.087)(4,676.919)(8,579.908)(16,829.715)};
  \addplot[dashed, cyan!60!black, mark=*, mark options={fill=white}, mark       
  size=1.5pt, line width=0.9pt, forget plot]                                    
    coordinates {(2,3374.203)(4,2121.370)(8,3401.809)(16,2329.569)};
  \addplot[dashed, brown!60!black, mark=triangle*, mark options={fill=white},   
  mark size=1.5pt, line width=0.9pt, forget plot]
    coordinates {(2,3183.548)(4,4202.233)(8,4514.583)(16,4392.524)};            
  \addplot[dashed, violet, mark=diamond*, mark options={fill=white}, mark       
  size=1.5pt, line width=0.9pt, forget plot]                                    
    coordinates {(2,1843.237)(4,2602.362)(8,2091.409)(16,1955.357)};            
                                                                                
  \end{axis}      
  \end{tikzpicture}    
  \vspace{-6.5mm}                                                            
  \caption{Scaling with the number of replicas for 1K and 10K events}                                 
  \label{fig:sinc_repl_sub}
  \end{subfigure}                                                               
                  
  \vspace{0.6em}                                                                
   
  \begin{subfigure}[t]{0.78\linewidth}
  \centering
  \begin{tikzpicture}
  \begin{axis}[
      width=\linewidth,
      height=5.3cm,                                                             
      xlabel={Number of events},
      ylabel={Avg. runtime (s)},                                                
      xmode=log,  
      log basis x=10,
      xtick={10,100,1000,10000},
      xticklabels={10,100,1K,10K},                                              
      ymode=log,
      yticklabel style={                                                        
          /pgf/number format/fixed,
          /pgf/number format/precision=3
      },
      grid=both,
      legend columns=2,                                                         
      legend style={
          font=\tiny,                                                           
          at={(0.02,0.98)},
          anchor=north west,
          draw=none,
          fill=white,
          fill opacity=1,
          text opacity=1,                                                       
          cells={anchor=west},
          inner sep=0.6pt,                                                      
          row sep=0.2pt,
      },
      legend image post style={scale=0.8},
      every axis plot/.append style={mark size=1.5pt, line width=0.9pt},
  ]                                                                             
   
  \addlegendimage{solid, black, mark=*, mark size=1.5pt, line width=0.9pt}
  \addlegendentry{Single}
  \addlegendimage{dashed, black, mark=*, mark options={fill=white}, mark
  size=1.5pt, line width=0.9pt}                                                 
  \addlegendentry{Incremental}
                                                                                
  \addplot[solid, blue, mark=*, mark size=1.5pt, line width=0.9pt]
    coordinates {(10,0.073)(100,0.105)(1000,2.104)(10000,384.430)};             
  \addlegendentry{ICS-ID 4R}                                                    
  \addplot[solid, red, mark=*, mark size=1.5pt, line width=0.9pt]               
    coordinates {(10,0.075)(100,0.082)(1000,3.343)(10000,695.349)};             
  \addlegendentry{SLS-ID 4R}                                                    
  \addplot[solid, cyan!60!black, mark=square*, mark size=1.5pt, line            
  width=0.9pt]                                                                  
    coordinates {(10,0.071)(100,0.088)(1000,2.912)(10000,455.928)};
  \addlegendentry{ICS-ID 8R}                                                    
  \addplot[solid, magenta, mark=square*, mark size=1.5pt, line width=0.9pt]     
    coordinates {(10,0.067)(100,0.091)(1000,4.467)(10000,654.806)};             
  \addlegendentry{SLS-ID 8R}                                                    
  \addplot[solid, brown!60!black, mark=triangle*, mark size=1.5pt, line
  width=0.9pt]                                                                  
    coordinates {(10,0.073)(100,0.085)(1000,1.830)(10000,300.134)};
  \addlegendentry{ICS-DD 4R}                                                    
  \addplot[solid, orange, mark=triangle*, mark size=1.5pt, line width=0.9pt]
    coordinates {(10,0.057)(100,0.451)(1000,555.965)};                          
  \addlegendentry{SLS-DD 4R}
  \addplot[solid, violet, mark=diamond*, mark size=1.5pt, line width=0.9pt]     
    coordinates {(10,0.081)(100,0.086)(1000,1.887)(10000,504.147)};
  \addlegendentry{ICS-DD 8R}                                                    
  \addplot[solid, black!50, mark=diamond*, mark size=1.5pt, line width=0.9pt]
    coordinates {(10,0.054)(100,0.446)(1000,554.633)};                          
  \addlegendentry{SLS-DD 8R}
                                                                                
  \addplot[dashed, blue, mark=*, mark options={fill=white}, mark size=1.5pt,    
  line width=0.9pt, forget plot]
    coordinates {(10,0.106)(100,0.220)(1000,6.860)(10000,3436.117)};
  \addplot[dashed, red, mark=*, mark options={fill=white}, mark size=1.5pt, line
   width=0.9pt, forget plot]                                                    
    coordinates {(10,0.066)(100,0.209)(1000,15.603)(10000,4824.000)};           
  \addplot[dashed, cyan!60!black, mark=square*, mark options={fill=white}, mark 
  size=1.5pt, line width=0.9pt, forget plot]                                    
    coordinates {(10,0.082)(100,0.219)(1000,9.443)(10000,2466.741)};            
  \addplot[dashed, magenta, mark=square*, mark options={fill=white}, mark       
  size=1.5pt, line width=0.9pt, forget plot]
    coordinates {(10,0.068)(100,0.190)(1000,14.307)(10000,4616.087)};           
  \addplot[dashed, brown!60!black, mark=triangle*, mark options={fill=white},   
  mark size=1.5pt, line width=0.9pt, forget plot]                               
    coordinates {(10,0.085)(100,0.194)(1000,8.890)(10000,2892.565)};            
  \addplot[dashed, orange, mark=triangle*, mark options={fill=white}, mark      
  size=1.5pt, line width=0.9pt, forget plot]
    coordinates {(10,0.078)(100,0.702)(1000,1017.490)};                         
  \addplot[dashed, violet, mark=diamond*, mark options={fill=white}, mark       
  size=1.5pt, line width=0.9pt, forget plot]                                    
    coordinates {(10,0.073)(100,0.172)(1000,7.233)(10000,2276.695)};            
  \addplot[dashed, black!50, mark=diamond*, mark options={fill=white}, mark     
  size=1.5pt, line width=0.9pt, forget plot]                                    
    coordinates {(10,0.077)(100,1.045)(1000,1058.401)};                         
                                                                                
  \end{axis}      
  \end{tikzpicture}                                                             
  \vspace{-2mm}   
  \caption{Scaling with the number of events}
  \label{fig:sinc_ops_sub}
  \end{subfigure}
  \vspace{1mm}                                                                  
  \caption{Single vs.\ incremental SLS and ICS for  
  isolate-delete (ID) and detach-delete (DD).}                               
  \label{fig:single_inc_all}                                                    
  \end{figure}
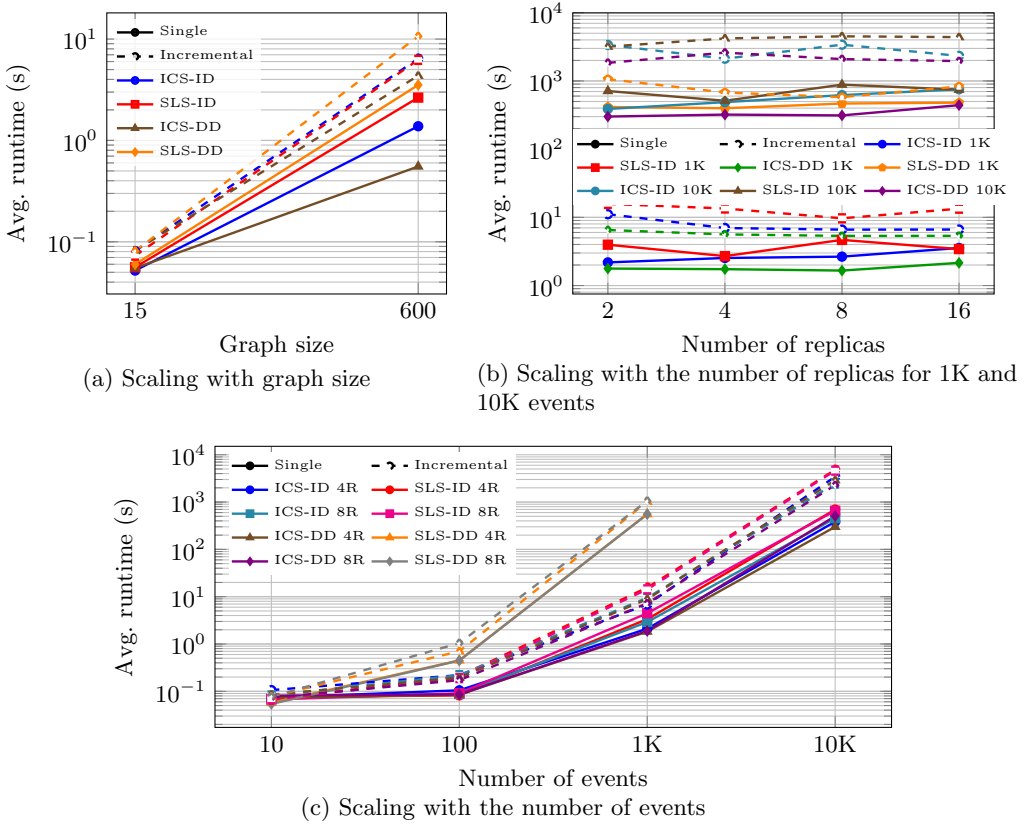    

%% file: conclusions.tex


We introduced \textsc{CRDTlog}, a declarative framework for specifying and reasoning about CRDT semantics over operation contexts, cleanly separating specification-level semantics (SLS) from implementation-level compositional semantics (ICS) via CRDT composition and transformation rules. Both are implemented in Datalog with modular support for basic, nested, and composite CRDTs.

Evaluation on two collaborative graph applications showed that ICS consistently matches SLS across all test configurations. Scalability is driven mainly by execution size, with replica count having limited impact; for larger workloads, ICS often outperforms SLS, making it well suited for large-scale testing and semantic prototyping.

As future work, we plan to formalize \textsc{CRDTlog} in a theorem prover to enable machine-checked proofs of additional semantic properties.

\section*{Acknowledgements}
This work was partially supported by the VERDI grant ANR-24-CE25-1109 (Dumbrava).




%% file: appendix.tex
\appendix

\section{Evaluation details}
\label{appendix:details}

This section reports additional details of the experimental evaluation summarized in Section \ref{experiments}. In particular, it provides the full quantitative results of the semantic equivalence tests between specification-level semantics (SLS) and implementation-level compositional semantics (ICS) for both graph variants. The table reports average execution times across configurations, complementing the scalability plots in the main text and supporting the empirical claims about correctness and performance.

\begin{table}[H]
\centering
\caption{Semantic Equivalence. Avg. execution time (s) per run, 1K runs/ category.} 
\begin{tabular}{lcccccccc}
\toprule
\#Ops & ICS ID& SLS ID & ICS DD & SLS DD & ICS ID& SLS ID& ICS DD& SLS DD\\
& 5R & 5R & 5R & 5R & 10R & 10R & 10R & 10R \\
\midrule
20   & 0.029  & 0.015 &0.018&0.015 & 0.023  & 0.012 &0.017&0.014 \\
50  & 0.018  & 0.011 &0.017&0.015  & 0.025  & 0.014  &0.018&0.016 \\
100 & 0.031  & 0.020 &0.019&0.017 & 0.025  & 0.016 &0.019&0.018 \\
1000  & 0.377 & 0.485 &0.343&0.531 & 0.454 & 0.501 &0.383&0.555\\
\bottomrule
\label{tab:se}
\end{tabular}
\end{table}

\section{Proof details}
\label{appendix:theorems}

This section presents formal proofs that both the SLS and ICS of the detach-delete graph preserve the no-dangling-edges invariant under all valid operation contexts.

\begin{lemma}
For the SLS of the detach-delete (DD) graph, the invariant of no dangling edges holds for any valid operation context.
\end{lemma}

\begin{proof}
Let $C=(\E,\op,\vis,\ar)$ be a well-formed operation context satisfying the input
preconditions of the DD graph $(V,A)$, i.e., \addE{$u$}{$v$} and \rmvN{$n$} are issued only if $n$, $u$, and $v$ are observed at the issuing replica. 
We show the resulting graph has no dangling edges.

A dangling edge can only arise due to a node-removal event, so we analyze executions
in which some \rmvN{$n$} occurs. We fix such a node $n$ and proceed by
case distinction. Throughout, we rely on \Cref{eq:graphddADTEdges}: an edge $(u,v)$ is present
only if it is added, not observed to be removed, and $\Gamma(e,u)\land \Gamma(e,v)$ holds
for the corresponding add-edge event $e$.

\begin{description}
  \item[Case 1.] The user removes node $n$, and there is no concurrent \addE{$n$}{\_} or \addE{\_}{$n$}.\\

For this case, we need to consider three subcases: if $n$ has no incident edges, if $n$ has outgoing edges, and if $n$ has incoming edges. 

\begin{itemize}
    \item $n$ has no incident edges.
    
 Then, for every $(u,v)$, we have $u\neq n$ and $v\neq n$, hence no edge can be
  incident to $n$ in the resulting state. Thus, the invariant holds.
  \item $n$ has outgoing edges. 
  
  Since $\rmvN{n}$ exists and is visible after the add-edge event along that branch, $\Gamma(e,n)$:
  $\neg\exists s\in E:\ \text{op}(s)=\mathsf{rmvN}(n)\land e\visrel s$,
  is false. This contradicts $(n,m)\in A$. Hence, no edge outgoing from $n$ is present in the resulting state, so no dangling edge arises. This holds for any outgoing edge candidate, $(n,m)$
  \item $n$ has incoming edges. 
  
  This subcase is analogous to the previous one. 
\end{itemize}

  \item[Case 2.] User $B$ removes node $n$, while concurrently user $C$ adds an edge
  incident to $n$ (i.e., $\addE{n}{m}$ or $\addE{m}{n}$).\\

  Let $s$ be a remove-node event with $\op(s)=\rmvN{n}$, and let $s'$
  be a concurrent add-edge event incident to $n$ as above. Then $\Gamma(e,n)$ evaluates to false for any edge to and from $n$ that is not $(n,m)$. For $(n,m)$ it would evaluate to true, as the \rmvN{$n$} is concurrent and not visible. 
  In this situation, however, the node predicate \Cref{eq:graphADT} also keeps
  $n$ present: if a visible $\rmvN{n}$ exists but there is a concurrent
  incident add-edge, then $\F{GraphDD\_SLS}(\hasN{n}, \ldots)=\mathsf{true}$. Hence, $n \in V$ and the edge $(n,m)$ will be in the resulting state, but none of the previous edges to and from $n$.T herefore, the edge $(n,m)$ that remains incident to $n$ is not dangling.\\
\end{description}

Finally, concurrent removals of the same node, removals of distinct nodes, or removals
interleaved with additions of unrelated nodes do not introduce additional ways of
creating an edge whose endpoint is absent: edge presence is always guarded by the
endpoint conditions $\Gamma(e,\cdot)$ in~\Cref{eq:graphddADTEdges} together with the
node predicate~\Cref{eq:graphADT}. Therefore, no resulting state from a operation context contains a dangling edge.
\end{proof}

\begin{lemma}
For the ICS of the detach-delete (DD) graph, the invariant of no dangling edges holds for any valid operation context.
\end{lemma}

\begin{proof}
Let $C=(\E,\op,\vis,\ar)$ be a well-formed operation context satisfying the input
preconditions of the DD graph (i.e., \addE{$u$}{$v$} and \rmvN{$n$} are issued only if $n$, $u$, and $v$ are observed at the issuing replica). 
We show that the resulting graph has no dangling edges.

A dangling edge can only arise due to a node-removal event, so we analyze executions where a \rmvN{n} occurs. We fix such a node $n$ and do a case analysis. We rely on \Cref{fig:transformation_dd} and \Cref{eq:graphddicsedge}: an edge $(u,v)$ is present only if it is in the edges set.

\begin{description}
  \item[Case 1.] The user removes node $n$ and there is  no
  concurrent $\addE{n}{\_}$ or $\addE{\_}{n}$.

We again consider three subcases: if $n$ has no incident edges, if $n$ has outgoing edges, and if $n$ has incoming edges. 

\begin{itemize}
    \item $n$ has no incident edges.
    
 Then, for every $(u,v)$, we have $u\neq n$ and $v\neq n$; hence, no edge can be
  incident to $n$ in the resulting state. Thus, the invariant holds.
  
  \item $n$ has outgoing edges.
  
    Each \rmvN{$n$} is translated to a \del{$n,m$} for each outgoing edge in the underlying set, thus removing all outgoing edges. As there is no concurrent operation that adds anything to the set, there is no \texttt{add} to win over the removal. Thus, the resulting state has no dangling edges. 
  
  \item $n$ has incoming edges.

  This case is analogous to the previous one; \rmvN{$n$} is also translated to a \del{$m,n$} for every incoming edge. 
\end{itemize}

  \item[Case 2.] The user removes node $n$, while concurrently, user $B$ adds an edge
  incident to $n$ (i.e., $\addE{n}{m}$ or $\addE{m}{n})$.

  \rmvN{$n$} translates to a \del{$n$} in the nodes set and \addE{$n$}{$m$} translates to \{\add{$n$}, \add{$m$}\}. In the concurrency conflict between \del{$n$} and \add{$n$}, the latter would win. Thus, the resulting state would contain $n$; hence, this case does not cause dangling edges.

  \item[Case 3.] User $B$ removes node $n$, while concurrently, user $C$ removes an
  edge between nodes $n$ and $m$.

  In this case, the removal of $n$ would be successful, as \rmvE{$n$}{$m$} does not have a transformation into the nodes set. This case would not cause any conflicts (add vs. remove) in the edges set CRDT either, because concurrent removals do not conflict. Thus, the resulting state will contain neither $n$ nor any dangling edges from or to $n$. 

\end{description}

Finally, concurrent removals of the same node, removals of distinct nodes, or removals
interleaved with additions of unrelated nodes do not introduce additional ways of
creating an edge whose endpoint is absent. Hence, no resulting state from an operation context can contain a dangling edge.    
\end{proof}